\documentclass[conference]{IEEEtran}
\IEEEoverridecommandlockouts

\usepackage{cite}
\usepackage{amsmath,amssymb,amsfonts}
\usepackage{algorithmic}
\usepackage{graphicx}
\usepackage{textcomp}
\usepackage{xcolor}
\def\BibTeX{{\rm B\kern-.05em{\sc i\kern-.025em b}\kern-.08em
    T\kern-.1667em\lower.7ex\hbox{E}\kern-.125emX}}

\usepackage[utf8]{inputenc}
\usepackage{amsthm}
\usepackage{physics}
\AtBeginDocument{\let\Pr\Probability}
\usepackage{hyperref}
\hypersetup{colorlinks=true,allcolors=black}
\usepackage[capitalize]{cleveref}

\newtheorem{theorem}{Theorem}
\newtheorem{proposition}{Proposition}
\newtheorem{lemma}{Lemma}
\newtheorem{corollary}{Corollary}
\newtheorem{conjecture}{Conjecture}

\def\calG{\mathcal{G}}
\def\calH{\mathcal{H}}
\def\calN{\mathcal{N}}

\def\CC{\mathbb{C}}

\DeclareMathOperator{\Per}{\mathrm{Per}}
\DeclareMathOperator{\Var}{\mathrm{Var}}
\DeclareMathOperator*{\Ex}{\mathbb{E}}
\def\vac{\mathrm{vac}}
\DeclareMathOperator{\poly}{poly}
\newcommand{\floor}[1]{\left\lfloor #1 \right\rfloor}

\def\sharpP{\mathsf{\#P}}
\def\BPP{\mathsf{BPP}}
\def\NP{\mathsf{NP}}
\def\P{\mathsf{P}}

\begin{document}

\title{Weak Permanent Anti-Concentration for Random Gaussian Matrices in Boson Sampling}

\IEEEoverridecommandlockouts

\IEEEpubid{\begin{minipage}{\textwidth}\ \\[48pt]
\copyright~2026 IEEE. Personal use of this material is permitted.  Permission from IEEE must be obtained for all other uses, in any current or future media, including reprinting/republishing this material for advertising or promotional purposes, creating new collective works, for resale or redistribution to servers or lists, or reuse of any copyrighted component of this work in other works.\end{minipage}}

\author{
\IEEEauthorblockN{Fei Meng\IEEEauthorrefmark{1}\IEEEauthorrefmark{2}\IEEEauthorrefmark{3}\textsuperscript{,+},
Bin Cheng\IEEEauthorrefmark{4}\textsuperscript{,+},
Jianan Li\IEEEauthorrefmark{3},
Man-Hong Yung\IEEEauthorrefmark{3}\IEEEauthorrefmark{5}}\\
\IEEEauthorblockA{
\IEEEauthorrefmark{1}University of Glasgow, UK\\
\IEEEauthorrefmark{2}City University of Hong Kong, China\\
\IEEEauthorrefmark{3} Southern University of Science and Technology, China\\
\IEEEauthorrefmark{4} National University of Singapore, Singapore\\
\IEEEauthorrefmark{5}International Quantum Academy, China\\
Email: fei.meng@glasgow.ac.uk, bincheng@nus.edu.sg, 11930493@mail.sustech.edu.cn, yung@iqasz.cn }
}

\maketitle
\bstctlcite{BSTcontrol}
\begingroup\renewcommand\thefootnote{+}
\footnotetext{These authors contributed equally to this work.}
\endgroup

\begin{abstract}
Recent demonstrations of quantum computational advantage have been driven largely by sampling problems.
A prominent model, boson sampling, involves sampling from the output distribution of a linear optical network.
However, its classical hardness hinges on two plausible yet less-studied conjectures: the average-case hardness of approximating Gaussian permanents, and the permanent anti-concentration conjecture (PACC).
The PACC is a purely mathematical assertion regarding the distributional properties of random Gaussian matrices.
While the typical magnitude of the permanent has been established for discrete random matrices, the complex Gaussian case, which governs transition amplitudes in linear optical networks, has remained open.
Here, we establish a weak anti-concentration bound by upper-bounding the probability that a random Gaussian permanent is superexponentially smaller than its standard deviation.
Tightening this bound to an inverse-polynomial fraction would prove the original PACC.
As a corollary, we establish the typical magnitude of Gaussian permanents, on par with Tao and Vu's seminal result for Bernoulli matrices.
Combined with the Aaronson-Arkhipov framework, our result implies that classically simulating boson sampling to within a superexponentially small total variation distance would collapse the polynomial hierarchy, assuming the remaining conjectures hold.
\end{abstract}


\section{Introduction}\label{sec1}

Quantum computational advantage~\cite{zhu2022quantum,deshpande2022quantum,madsen2022quantum,wu2021strong,zhong2020quantum}, also known as quantum supremacy~\cite{arute2019quantum,boixo2018characterizing,lund2017quantum,terhal2018quantum,villalonga2020establishing,harrow2017quantum,hangleiter2023computational}, is the capability of quantum computing devices to tackle problems that classical computers practically cannot. Several models have been proposed and experimentally explored, including boson sampling~\cite{Aaronson2013,tillmann2013experimental,hamilton2017gaussian, brod2019photonic,spring2013boson,wang2017high,lund2014boson,bentivegna2015experimental,gard2015introduction}, IQP sampling~\cite{shepherd_temporally_2009,bremner2011classical,bremner_average-case_2016,bremner_achieving_2017,bremner2025instantaneous}, and random circuit sampling~\cite{arute2019quantum,zlokapa2023boundaries,bouland2019complexity,zhu2022quantum}.  

Boson sampling, a quantum computational scheme designed for non-interacting bosons, e.g., photons, was initially proposed by Aaronson and Arkhipov as a rudimentary model of quantum computing to provide theoretical evidence of quantum computational advantage ~\cite{Aaronson2013}.   Boson sampling demands less coherent quantum control, requiring only passive linear-optical elements like beam splitters and phase shifters, which makes it experimentally accessible with current photonic technology.
Moreover, Gaussian boson sampling, a variant of boson sampling, finds applications in graph-related problems~\cite{deng2023solving}, such as, sampling dense subgraphs~\cite{arrazola2018using} and finding perfect matchings~\cite{bradler2018gaussian}.

Despite many successful experimental
demonstrations~\cite{brod2019photonic,wang2018toward,wang2017high,spagnolo2014experimental,zhong2019experimental,bentivegna2015experimental,tillmann2013experimental, spring2013boson, wang2019boson, loredo2017boson, broome2013photonic}, the theoretical foundations of boson sampling remain incomplete.
The classical hardness of boson sampling relies on two unproven conjectures, namely, the Gaussian permanent estimation conjecture and the permanent anti-concentration conjecture (PACC)~\cite{Aaronson2013}.
Assuming these two conjectures, one can prove that classically sampling from the output distribution of a boson sampling experiment, even approximately, is intractable, unless the polynomial hierarchy collapses, which is highly implausible.
Recently, the hardness framework has been extended from the dilute regime ($m \gg n^2$) to the experimentally relevant saturated regime ($m = \Theta(n)$), where output collisions are frequent~\cite{bouland2023complexity}; the anti-concentration conjecture remains open in all regimes.

Several approaches to PACC have been pursued.
Aaronson and Arkhipov proved a weak version of anti-concentration by computing the first and second moments of the squared permanent of Gaussian matrices, and suggested that computing higher moments might eventually prove PACC~\cite{Aaronson2013}. Subsequently, Nezami used a hybrid representation-theoretic and combinatorial approach to analyze and lower bound the higher moments~\cite{nezami_permanent_2021}.
Although the lower bounds may be tight, they do not directly imply the PACC, and it has been argued that additional assumptions are necessary to prove the conjecture.
More recently, essentially the same second-moment argument of Aaronson and Arkhipov has been adapted to the saturated regime, yielding analogous weak anti-concentration bounds with the same asymptotic scaling~\cite{mhiri2026boson,kolarovszki2026general}.
For Gaussian boson sampling (GBS), Ehrenberg el al. relate anti-concentration to the ratio of the squared first moment to the second moment of the output probability, where the output probability is determined by the Hafnian, a generalization of the permanent~\cite{ehrenberg_second_2025,ehrenberg_transition_2025}. 
They then develop a graph-theoretic method to compute the first moment and analyze certain properties of the second moment, which enables them to identify a transition from no anti-concentration to weak anti-concentration.
Given the relationship between permanents and Hafnians~\cite{barvinok_combinatorics_2016}, these techniques may also be applicable to standard boson sampling.

In a pure mathematical context, Tao and Vu proved that the permanent of a random Bernoulli matrix has magnitude $n^{(1/2+o(1))n}$ asymptotically almost surely~\cite{tao2009permanent}.
Their technique has been subsequently improved and extended to symmetric and more general discrete random matrices~\cite{kwan_permanent_2022,hunter_exponential_2025}. However, these results concern discrete distributions and do not directly apply to the complex Gaussian matrices relevant to boson sampling.
The complex Gaussian distribution arises naturally in this context: the permanent of a submatrix of a Haar-random unitary matrix determines the transition amplitude of identical photons through a linear optical network~\cite{Aaronson2013}, and when the number of optical modes $m$ greatly exceeds the number of photons $n$, these submatrices are well approximated by i.i.d.\ complex Gaussian matrices~\cite[Theorem~35]{Aaronson2013}. Anti-concentration for this distribution is thus needed to complete the hardness argument for boson sampling.

Recently, assuming PACC, Bouland et al. proved that estimating the output probabilities of random boson sampling experiments to additive error $e^{-n\log n - n - O(n^{\delta})}$ is $\sharpP$-hard~\cite{bouland2025exponential}, substantially narrowing the gap to the Gaussian permanent estimation conjecture.

Here, we take a complementary approach and prove a weak version of permanent anti-concentration by extending the row-exposure framework of Tao and Vu~\cite{tao2009permanent} to random Gaussian matrices.
This extension was suggested in Aaronson and Arkhipov's paper~\cite{Aaronson2013}, but the technical details have not been worked out since then.
The passage from Bernoulli to complex Gaussian requires replacing several tools in the proof:
(i) the Littlewood-Offord-Erd\H{o}s inequality is replaced by an anti-concentration bound for complex Gaussian linear combinations (\cref{littlewoodinequality});
(ii) the combinatorial lower bound on the probability of growing a minor is replaced by an analytic argument exploiting the rotational symmetry of the complex Gaussian distribution (\cref{lemma:large_parent_has_large_child});
and (iii) the Azuma-Hoeffding inequality, which requires bounded increments, is replaced by the McDiarmid inequality to handle the unbounded support of Gaussian variables (\cref{lemma:McDiarmid}).
Furthermore, combined with the framework of Aaronson and Arkhipov~\cite{Aaronson2013}, our weak anti-concentration bound implies that classical simulation of boson sampling to within superexponentially small total variation distance ($O(1/n^{2\alpha n})$) would collapse the polynomial hierarchy, under the Permanent-of-Gaussians Conjecture (\cref{remark:hardness}).

In what follows, we first formulate the conjecture in \cref{sec:preliminary}, and then present our main analytical result in \cref{sec:results}. Then we present the sketch of the proof in \cref{sec:sketch_of_proof}. Finally, we discuss the implications of our results and conclude.

\section{Boson sampling and permanent}
\label{sec:preliminary}

The permanent of an $n \times n$ matrix $M = (m_{ij})$ is defined as the sum over all permutations $\pi \in S_n$ of the product of matrix elements $m_{i,\pi(i)}$:
\begin{equation}
    \Per(M) = \sum _ {\pi \in S_n } \prod_i m_{i,\pi(i)}\, .
\end{equation}
It plays an important role in combinatorics, statistics, and quantum physics.
In quantum physics, it is related to the transition amplitude of identical photons from one configuration to another, under the action of a linear optical network.

Consider a linear optical network with $m$ modes and $n$ photons.
We use $\vb{s} = (s_1, \cdots, s_m)$ to denote a photon configuration of the linear optical network, where $s_i$ denotes the number of photons in the $i$-th mode.
The corresponding Fock state is $\ket{\vb{s}} := \ket{s_1} \otimes \cdots \otimes \ket{s_m}$, where $\ket{s_i} := (s_i!)^{-1/2} (a_i^{\dagger})^{s_i} \ket{\vac}$ and $a_i^{\dagger}$ is the bosonic creation operator for the $i$-th mode.
Since photons are neither created nor destroyed in the linear optical network, the set of possible configurations is $\Phi_{m, n} := \{ (s_1, \cdots, s_m) : \sum_{i=1}^m s_i = n\}$ and the associated Hilbert space is denoted as $\calH_{m, n}$.
The general quantum state in this Hilbert space is given by a unit vector
\begin{align}
    \ket{\psi} = \sum_{\vb{s} \in \Phi_{m, n}} \alpha_{\vb{s}} \ket{\vb{s}} \ ,
\end{align}
satisfying $\sum_{\vb{s} \in \Phi_{m, n}} |\alpha_{\vb{s}}|^2 = 1$.

The linear optical network can be modelled by an $m\times m$ unitary matrix $U = (u_{ij})$. 
Its unitary transformation in the Hilbert space $\calH_{m, n}$, denoted as $\varphi(U)$, acts as a linear transformation on the bosonic modes,
\begin{align}
    \varphi(U) a_i^{\dagger} \varphi(U)^{\dagger} = \sum_{j=1}^m u_{ij} a_j^{\dagger} \ .
\end{align}
Under the action of $U$, the transition amplitude from $\ket{\vb{t}}$ to $\ket{\vb{s}}$ is given by~\cite[Theorem~3.10]{Aaronson2013}
\begin{align}
    \mel{\vb{s}}{\varphi(U)}{\vb{t}} = \frac{\Per(U_{\vb{s}, \vb{t}})}{\sqrt{s_1! \cdots s_m! t_1! \cdots t_m!}} \ ,
\end{align}
when $\vb{s}, \vb{t} \in \Phi_{m, n}$.
Here, $U_{\vb{s}, \vb{t}}$ is an $n\times n$ submatrix of $U$ constructed in the following way.
First, take $s_i$ copies of the $i$-th rows of $U$ to form $U_{\vb{s}}$.
Then, take $t_j$ copies of the $j$-th row of $U_{\vb{s}}$ to form $U_{\vb{s}, \vb{t}}$.

For boson sampling, the standard initial state $\ket{\vb{t}}$ is given by $\ket{1_n} := \ket{1, \cdots, 1, 0, \cdots, 0}$, which is the configuration with a single photon in the first $n$ modes and no photon in the others.
Then, the probability of observing the configuration $\vb{s} = (s_1, \cdots, s_m)$ is
\begin{align}
\label{eq:boson_sampling_probability}
    \Pr[\vb{s}] = |\alpha_{\vb{s}}|^2 = \frac{|\Per(U_{\vb{s}, 1_n})|^2}{s_1! \cdots s_m!} \ .
\end{align}
The problem of boson sampling is to sample from the above probability distribution, given a unitary matrix $U$ as input, which describes a linear optical network.

When the linear optical network is randomly constructed, one expects that $U$ is approximately Haar random. 
Moreover, it is argued that after scaled up by a factor $\sqrt{m}$, the submatrix $U_{\vb{s}, 1_n}$ taken from a Haar random $U \in \mathcal{H}_{m, m}$ is approximately a matrix sampled from $\calG^{n \times n}$~\cite[Sec.~5.1]{Aaronson2013}.
Here, $\calG$ denotes the standard complex Gaussian distribution with mean 0 and variance $\Ex_{z \sim \calG} [|z|^2] = 1$, whose real and imaginary parts are independent and each follow a normal distribution with mean $0$ and standard deviation $1/2$.
A matrix from $\calG^{n \times n}$, called \emph{Gaussian matrix}, has its entries sampled i.i.d. from $\calG$.
For this reason, below we consider a matrix $M \sim \calG^{n \times n}$ instead of $U_{\vb{s}, 1_n}$.
By direct calculation, we can show the expectation value and the variance of $\Per(M)$ when $M$ is a random Gaussian matrix,
\begin{gather}
    \Ex\limits_{M \sim \calG^{n\times n}}[\Per(M)] = 0 \label{eq:perm_exp} \\
    \Var(\Per(M)) = \Ex\limits_{M \sim \calG^{n\times n}}[|\Per(M)|^2] = n! \ .\label{eq:perm_var}
\end{gather}

The definition of permanent is similar to that of the determinant, except that all the signs in the summation are positive.
Despite its seemingly simpler definition, the computation of the permanent is actually $\sharpP$-hard~\cite{valiant1979complexity}, which means that it is unlikely to be efficiently calculated on a classical computer, in contrast to the determinant.
Naively, this computational complexity suggests that sampling from the probability distribution of a linear optical network may be classically hard.
On the other hand, one can efficiently  sample from the probability distribution defined in \cref{eq:boson_sampling_probability} on a quantum computer by running the corresponding linear optical network and performing photon-number measurements.
This makes boson sampling a natural candidate for demonstrating quantum computational advantage.

More precisely, the classical hardness of boson sampling relies on the following three conjectures~\cite{Aaronson2013}. 
Namely, (a) the Polynomial Hierarchy does not collapse; (b) it is $\sharpP$-hard to approximate the permanent of Gaussian matrices on average; (c) the permanent of Gaussian matrices satisfies the anti-concentration bound.
While the non-collapse of Polynomial Hierarchy is widely accepted, the last two conjectures are less-studied.
In this work, we focus on the anti-concentration conjecture, whose mathematical statement is as follows.

\begin{conjecture}[Permanent anti-concentration~\cite{Aaronson2013}]
\label{conj:PACC}
    There exists a polynomial $p$, such that for all $n$ and $\delta >0$,
    \begin{equation}
        \Pr_{M \sim \calG^{n \times n}} \left( |\Per(M)| < \frac{\sqrt{n!}}{p(n,1/\delta)} \right) < \delta \, .
    \end{equation}
\end{conjecture}

Equivalently, \cref{conj:PACC} states that there exist positive constants $C, D$, and $\beta$, such that for all $n$ and $\epsilon > 0$,
\begin{equation}\label{eq:PACC_original}
\underset{M \sim \calG^{n \times n}}{\Pr}
\left( |  \Per (M) | < \epsilon {\sqrt{n!}} \right) < C n^D \epsilon^\beta \, .
\end{equation}

\section{Weak anti-concentration bound}\label{sec:results}

In this work, we extend the results reported in Ref.~\cite{tao2009permanent} for random Bernoulli matrices to derive a weak anti-concentration bound for random Gaussian matrices. 
Specifically, we prove the following weak permanent anti-concentration theorem.

\begin{theorem}[Weak permanent anti-concentration]
\label{thm:weakPACC}
Let $M \sim \calG^{n \times n}$ be a random Gaussian matrix.
Then, there exists a constant $c > 0$ such that for all $\alpha > 0$ and $n$ sufficiently large depending on $\alpha$,
\begin{equation}
\Pr_{M \sim \calG^{n \times n}} \left( | \Per (M) | < \frac{\sqrt{n!}}{n^{\alpha n}} \right) =
O \left( \frac{1}{n^c} \right) \, .
\end{equation}
\end{theorem}

Compared to \cref{conj:PACC}, \cref{thm:weakPACC} is a weak version because it only upper bounds the probability that $|\Per(M)|$ is smaller than $1/n^{\alpha n}$ fraction of its standard deviation (i.e., $\sqrt{n!}$), instead of a polynomially small fraction.
Since $1/n^{\alpha n}$ is asymptotically smaller than $1/\poly(n, 1/\delta)$, the cumulative probability bound in \cref{thm:weakPACC} does not suffice to imply the bound required for PACC.

A corollary of \cref{thm:weakPACC} is the typical value of the permanent of a random Gaussian matrix.
Using Stirling's approximation, we have $n! = n^{(1 + o(1)) n}$.
Then, we have the following corollary.

\begin{corollary}[Typical value of Gaussian permanents]
\label{coro:typicalvalue}
    For a random matrix $M \sim \calG^{n \times n}$, we have 
    \begin{equation}
        | \Per (M) | = n^{( \frac{1}{2} + o (1) ) n} \, ,
    \end{equation}
    asymptotically almost surely, i.e., with probability $1 - o (1)$.
\end{corollary}

\begin{proof}
    An equivalent formulation of the corollary is to show that for all $\alpha > 0$, we have
    \begin{align}
        \Pr(n^{(1/2-\alpha) n} \leq |\Per(M)| \leq n^{(1/2+\alpha) n}) = 1 - o(1) \, .
    \end{align}
    \cref{thm:weakPACC} shows that the probability of $|\Per(M)| \geq n^{(1/2 + o(1) - \alpha) n}$ is $1 - o(1)$ for all $\alpha > 0$, which gives the side of the lower bound.
    To establish the upper bound, one can use the Chebyshev's inequality.
    Recall from \cref{eq:perm_exp,eq:perm_var} that the expectation and the variance of $\Per(M)$ are $0$ and $n!$, respectively.
    Using Chebyshev's inequality, one can show that the probability of $|\Per(M)| \leq k \sqrt{n!}$ is at least $1 - 1/k^2$.
    Taking $k$ to be $n^{\alpha n}$ proves the upper bound.
\end{proof}

\subsection{Implication for classical hardness of boson sampling}
\label{remark:hardness}

Our weak anti-concentration bound can be combined with the framework of Aaronson and Arkhipov~\cite{Aaronson2013} to yield a conditional hardness result, which we sketch in the following.
By Theorem~1.3 of Ref.~\cite{Aaronson2013}, if a classical algorithm samples from the boson sampling distribution to within total variation distance $\delta$, then $|\Per(M)|^2$ can be estimated to within additive error $O(\delta \cdot n!)$ in $\BPP^{\NP}$, which is within the third level of polynomial hierarchy.
\cref{thm:weakPACC} guarantees $|\Per(M)|^2 \ge n!/n^{2\alpha n}$ with probability $1 - O(n^{-c})$, so this additive error translates to a multiplicative error of $O(\delta \cdot n^{2\alpha n})$ on $|\Per(M)|^2$.
For the multiplicative error to be at most constant, one needs $\delta = O(1/n^{2\alpha n})$. Following the triangulation argument of Theorem~1.7 in Ref.~\cite{Aaronson2013}, this yields an $O(1)$ multiplicative error estimate of $\Per(M)$ in $\BPP^{\NP}$.
Under the Permanent-of-Gaussians Conjecture~\cite{Aaronson2013}, which asserts that $\poly(n)$ multiplicative error estimation of $\Per(M)$ is $\sharpP$-hard on average, one concludes $\P^{\sharpP} \subseteq \BPP^{\NP}$.
By the Toda's theorem~\cite{toda_pp_1991}, this implies the collapse of the polynomial hierarchy.

In summary, classical simulation of boson sampling to within $O(1/n^{2\alpha n})$ total variation distance, combined with the Permanent-of-Gaussians Conjecture, would collapse the polynomial hierarchy. This accuracy threshold is far more stringent than the $1/\poly(n)$ that would follow from the full PACC.

\section{Sketch of the proof}
\label{sec:sketch_of_proof}

To prove \cref{thm:weakPACC}, we adopt the row-exposure strategy of Tao and Vu~\cite{tao2009permanent}.
Specifically, we reveal the rows of the Gaussian matrix $M$ one at a time and track how the permanents of small submatrices grow into a lower bound on $\Per(M)$ itself.
In what follows, we will sketch the proof and leave the details to Appendices.

The key mechanism we rely on is the cofactor expansion.
When the $(k+1)$-th row $\vec{X}_{k+1}$ is exposed, the permanent of a $(k+1) \times (k+1)$ minor $M_{A \cup \{i\}}$ decomposes as
\begin{small}
\begin{equation}\label{eq:cofactor_sketch}
    \Per(M_{A \cup \{i\}}) = m_{k+1, i} \Per(M_A) + \sum_{j \in A} m_{k+1, j} \Per(M_{A \cup \{i\} \setminus \{j\}}) \, ,
\end{equation}
\end{small}
where $m_{k+1,i}$ and $m_{k+1,j}$ are i.i.d.\ entries from $\calG$ and $A \subset [n] := \{1,2,\ldots,n\}$ has size $k$.
Conditioned on the first $k$ rows, the right-hand side is a complex Gaussian random variable whose variance is at least $|\Per(M_A)|^2$.
By the rotational symmetry of the complex Gaussian distribution, the probability that $|\Per(M_{A\cup\{i\}})| \geq |\Per(M_A)|$ is at least $1/e$  (\cref{lemma:large_parent_has_large_child}).
This is the elementary step: each row exposure has a constant probability of producing a child minor whose permanent is at least as large as its parent's.

We now set up the notation for the induction.
Let $M^{(k)}$ denote the first $k$ rows of $M$.
The $k \times k$ \emph{minor} $M_A$ is the submatrix formed by the first $k$ rows and columns indexed by $A$.
We say $A$ is \textit{$\lambda$-heavy} if $|\Per(M_A)| \geq \lambda$, and let $E_{k, N, \lambda}$ denote the event that at least $N$ of the $\binom{n}{k}$ minors of size $k$ are $\lambda$-heavy. (Note that during our iterative growth steps, the parameter $N$ may take non-integer real values, such as $\epsilon N_k / 6$. In such cases, $E_{k, N, \lambda}$ mathematically requires the integer count of such minors to be at least $\lceil N \rceil$. We permit $N$ to be a real number purely to streamline the algebraic recurrences without cluttering the notation with ceiling functions at every step.)
By \cref{prop:prob_z_circle}, each $1 \times 1$ minor $m_{1,i}$ satisfies $|m_{1,i}| \geq 1$ with probability $1 - 1/\eta$, where $\eta := e/(e-1) \approx 1.582$, so
\begin{equation}\label{eq:initial_event}
    \Pr(E_{1, 1, 1}) > 1 - \eta^{-n} \, .
\end{equation}
To prove \cref{thm:weakPACC}, it suffices to show
\begin{equation}\label{eq:bound_event_prob}
    \Pr(E_{n, 1, \sqrt{n!}/n^{\alpha n}}) \geq 1 - O(n^{-\Omega(1)}) \, ,
\end{equation}
for all $\alpha > 0$ and $n$ sufficiently large.

The proof grows from \cref{eq:initial_event} to \cref{eq:bound_event_prob} by inductively increasing $k$ from $1$ to $n$.
At each step, we track two quantities: the number $N$ of $\lambda$-heavy minors and the magnitude threshold $\lambda$.
Two lemmas control the induction.
The first maintains the count of heavy minors as $k$ increases:

\begin{lemma}[Maintaining many large minors]
\label{maintainingmanylargeminors_main}
    Let $1 \le k \le (1 - \epsilon)n$ for some $\epsilon > 0$, $N \ge 1$, and $\lambda > 0$. Then
    \begin{equation}
        \Pr(E_{k+1,\, \epsilon N/6,\, \lambda} \mid E_{k,N,\lambda}) \ge 1 - 2\exp(-\Omega(\epsilon n)) \, .
    \end{equation}
\end{lemma}

The second lemma allows us to improve either $N$ or $\lambda$:

\begin{lemma}[Growing many large minors]
\label{lemma:growingmanylargeminors_main}
    Under the same conditions as \cref{maintainingmanylargeminors_main}, with an additional parameter $0 < c < 1$, we can partition $E_{k,N,\lambda}$ into two exclusive events $E'_{k,N,\lambda,c}$ and $E''_{k,N,\lambda,c}$ (depending only on the first $k$ rows) such that
    \begin{equation}
        \Pr(E_{k+1,\, n^c N,\, \lambda} \mid E'_{k,N,\lambda,c}) \ge \frac{1}{3}
    \end{equation}
    and
    \begin{equation}
        \Pr(E_{k+1,\, \epsilon N/4,\, n^{1/2-c}\lambda} \mid E''_{k,N,\lambda,c}) \ge 1 - n^{-c/4} \, .
    \end{equation}
\end{lemma}

Under the event $E'$, the count of heavy minors jumps from $N$ to $n^c N$ (population growth). Under $E''$, the magnitude threshold improves by a factor $n^{1/2-c}$ at the cost of a constant-factor reduction in $N$ (magnitude growth).
The proofs of both lemmas are given in \cref{proofofproposition3.12}; the main inputs are the cofactor expansion \cref{eq:cofactor_sketch}, the Gaussian Littlewood-Offord bound (\cref{littlewoodinequality}), and a double-counting argument.

The analysis proceeds in three steps, with $k_0 := \lfloor \epsilon n \rfloor$ and $k_1 := \lfloor (1-\epsilon)n \rfloor$.

\paragraph{Step 1 ($k = 1 \to k_0$): Warm-up.}
Using only the elementary $1/e$ bound, we show that at least one $1$-heavy minor survives up to $k = k_0$ (cf.\ \cref{prop:step_1}):
\begin{equation}
    \Pr(E_{k_0,1,1}) \ge 1 - e^{-\Omega(n)} \, .
\end{equation}
This step does not grow $N$ or $\lambda$; it provides the initial condition for Step~2.

\paragraph{Step 2 ($k_0 \to k_1$): Growth phase.}
At each step $k \to k+1$, we apply \cref{lemma:growingmanylargeminors_main} to attempt either population growth or magnitude growth, falling back on \cref{maintainingmanylargeminors_main} if neither succeeds.
The key observation is that population growth and magnitude growth \emph{cannot both stall for long}: if $N$ stays small, then the event $E''$ is triggered and $\lambda$ must grow; conversely, if $\lambda$ stays small, then $E'$ is triggered and $N$ must grow.
This tradeoff is made precise via a potential function $W_k$ that decreases whenever either type of growth occurs, and is shown to be a supermartingale with bounded increments (\cref{prop:success_prob}).
After $k_1 - k_0 \approx (1-2\epsilon)n$ steps, the accumulated magnitude growth yields (cf.\ \cref{proposition:second_step})
\begin{equation}
    \Pr(E_{k_1, 1, n^{(1/2-\alpha)n}} \mid E_{k_0,1,1}) \ge 1 - e^{-\Omega(n)} \, .
\end{equation}

\paragraph{Step 3 ($k_1 \to n$): Endgame.}
When $k > k_1 = (1-\epsilon)n$, different $(k+1)$-minors share most of their columns, making their permanents highly correlated. The row-by-row independence argument used in Steps~1--2 no longer applies.
Instead, we find many $(n-L)$-minors (with $L = O(\log n)$) whose complements are disjoint, and iteratively grow them to size $n$ by a separate induction (\cref{endgame}).
The McDiarmid inequality (\cref{lemma:McDiarmid}) is used here in place of the Azuma-Hoeffding inequality, because the Gaussian entries are unbounded and the standard bounded-difference condition fails.
The endgame incurs a multiplicative loss of $n^{\log n}$ in the permanent lower bound, which is absorbed into the $o(1)$ term in the exponent:
\begin{align}
    \Pr(E_{n, 1, n^{(1/2-\alpha-o(1))n}} \mid E_{k_1, 1, n^{(1/2-\alpha)n}}) \geq 1 - n^{-\Omega(1)} \, .
\end{align}
Combining all three steps proves \cref{eq:bound_event_prob}.

\section{Conclusion and Discussion}\label{sec:conclusion}

In this work, we prove a weak version of the permanent anti-concentration conjecture (PACC) for random complex Gaussian matrices, establishing that $|\Per(M)| = n^{(1/2+o(1))n}$ asymptotically almost surely (\cref{coro:typicalvalue}).
To the best of our knowledge, this is the first rigorous result on the typical magnitude of the permanent for the complex Gaussian distribution.
Previous results of this type were limited to discrete distributions: Bernoulli matrices~\cite{tao2009permanent}, symmetric Rademacher matrices~\cite{kwan_permanent_2022}, and more general discrete random matrices~\cite{hunter_exponential_2025}.
The complex Gaussian distribution is the physically relevant one for boson sampling, since submatrices of Haar-random unitary matrices --- which describe random linear optical networks --- are well approximated by i.i.d.\ complex Gaussian matrices in the regime $m \gg n^2$~\cite[Theorem~35]{Aaronson2013}. Through this approximation, our result implies that the transition amplitudes of $n$ photons through a random $m$-mode linear optical network have magnitude $n^{(1/2+o(1))n}/m^{n/2}$ with high probability, consistent with the $n!$ multi-photon path amplitudes adding as a random walk in the complex plane.

As discussed in \cref{remark:hardness}, the weak anti-concentration bound, combined with the arguments of Theorems~1.3 and~1.7 in Ref.~\cite{Aaronson2013}, implies that classical simulation of boson sampling to within $O(1/n^{2\alpha n})$ total variation distance would place $\sharpP$-hard problems in $\BPP^{\NP}$, collapsing the polynomial hierarchy. This accuracy threshold is far stronger than the $1/\poly(n)$ implied by the full PACC, and therefore our result does not yet suffice to establish the classical hardness of approximate boson sampling at physically relevant accuracy.

In \cite{Aaronson2013}, a weak anti-concentration bound of the form $\Pr_{M \sim \calG^{n\times n}} \left[ |\Per(M)|^2 < \alpha \cdot n! \right] < 1 - \frac{(1-\alpha)^2}{n+1}$ was proved, where $\alpha < 1$ is a constant.
This bound was proved by analyzing the second moment of $|\Per(M)|^2$ or equivalently, the output probability and then apply the Paley-Zygmund inequality.
Recently, similar second-moment analysis has been extended to the saturated regime, where $m = \Theta(n)$, and leads to a weak bound of similar scaling~\cite{mhiri2026boson,kolarovszki2026general}. 
One can see that such a bound is weaker than the Tao-Vu type superexponential bound.

Closing the gap between our weak bound (with denominator $n^{\alpha n}$, superexponential in $n$) and the full PACC (with denominator $\poly(n)$) remains an important open problem.
Our proof extends the row-exposure framework of Tao and Vu~\cite{tao2009permanent}, which inherently loses a factor of $n^{\log n}$ in the endgame step (\cref{endgame}). New techniques may be needed to overcome this barrier.
Whether the moment-based approach of Nezami~\cite{nezami_permanent_2021}, which provides tight lower bounds on all positive integer moments of $|\Per(M)|^2$, can be combined with our approach to yield stronger bounds is an open question.
Proving anti-concentration directly for submatrices of Haar-random unitary matrices, bypassing the Gaussian approximation, would extend the result to arbitrary ratios $m/n$ and strengthen the connection to boson sampling experiments.

\section*{Acknowledgments}
We thank the helpful discussion with Qin Li. F. M. acknowledges the support by City University of Hong Kong (Project No. 9610623) and the YTJX Academy. B.C. acknowledges the support by the CQT Young Researcher Career Development Grant (26-YRCDG-BC). M.H.Y. is supported by National Natural Science Foundation of China (11875160 and U1801661), 
the Natural Science Foundation of Guangdong Province (2017B030308003), 
the Science, Technology and
Innovation Commission of Shenzhen Municipality (JCYJ20170412152620376 and JCYJ20170817105046702 and KYTDPT20181011104202253), 
the Key R\&D Program of Guangdong province (2018B030326001), the Economy, 
Trade and Information Commission of Shenzhen Municipality (201901161512), 
Guangdong Provincial Key Laboratory (Grant No. 2019B121203002).

\appendices

\section{Preliminaries}

\subsection{Complex Gaussian random variables}

Here, we collect some results about complex Gaussian distributions.
A complex Gaussian random variable $z \sim \CC \calN(\mu, \sigma^2, C)$, is specified by three parameters, the \emph{mean}, the \emph{variance} and the \emph{relation},
\begin{align}
    \mu &:= \Ex [z] & \sigma^2 &:= \Ex[(z - \mu)(z - \mu)^*] & C &:= \Ex [(z - \mu)^2]\ .
\end{align}
The real and imaginary parts of a complex Gaussian random variable are bivariate Gaussian distributed.
Throughout this work, we only consider the case $C = 0$, which implies that the real and the imaginary part of $z$ are independent and of variance $\frac{\sigma^2}{2}$.
In this case, we can simplify the notation to be $z \sim \CC \calN (\mu, \sigma^2)$.
The standard complex Gaussian distribution is $\calG := \CC \calN (0, 1)$, whose real and imaginary parts are independent Gaussian random variables with mean 0 and variance $1/2$.

We first give the probability density function of a complex Gaussian random variable.

\begin{proposition}
    The probability density function for $z \sim \CC \calN (\mu, \sigma^2)$ is given by
    \begin{align}
        p(z) = \frac{1}{\pi \sigma^2} \exp(- \frac{|z - \mu|^2}{\sigma^2}) \ .
    \end{align}
\end{proposition}

\begin{proof}
    Let $z = x + iy$ and $\mu = \mu_x + i\mu_y$.
    We have $x \sim \calN(\mu_x, \sigma^2/2)$ and $y \sim \calN(\mu_y, \sigma^2/2)$ are independent Gaussian random variables.
    Multiplying the two independent Gaussian densities gives
    \begin{equation}
        p(z) = p(x)p(y) = \tfrac{1}{\pi \sigma^2} \exp\!\left(- \tfrac{|z - \mu|^2}{\sigma^2}\right) \, .
    \end{equation}
\end{proof}

Then, we can compute the probability that a complex Gaussian random variable deviates from its mean.

\begin{proposition}
    \label{prop:prob_z_circle}
    Let $z \sim \CC \calN(\mu, \sigma^2)$ be a complex Gaussian random variable with mean $\mu$ and variance $\sigma^2$. Then,
    \begin{align}
        \Pr (|z - \mu| \leq R) = 1 - e^{-\frac{R^2}{\sigma^2}} \ .
    \end{align}
    In particular, when $z \sim \calG$, we have $\Pr(|z| \leq 1) = 1 - e^{-1}$.
\end{proposition}

\begin{proof}
    Let $z' := z - \mu \sim \CC \calN(0, \sigma^2)$. Polar change of variables gives
    \begin{equation}
        \Pr (|z'| \leq R) = \int_{0}^{R} \tfrac{2 r}{\sigma^2} \exp\!\left(- \tfrac{r^2}{\sigma^2}\right) \dd{r} = 1 - e^{-R^2/\sigma^2} \, .
    \end{equation}
\end{proof}

\subsection{Probabilistic Tools}

The following lemma replaces the Littlewood-Offord-Erd\H{o}s inequality~\cite{erdos1945lemma} used in Ref.~\cite{tao2009permanent}.
\begin{lemma}
  \label{littlewoodinequality} Let $\lambda > 0$, and $m, k \geq 1$. Let $v_1, \ldots, v_m$ be complex numbers such that at least $k$ of them have modulus larger than $\lambda$. Then for $a_1,
  \ldots, a_m$ i.i.d. sampled from $\calG$, we have
  \begin{equation}
    \Pr ( | a_1 v_1 + \cdots + a_m v_m | \le r \lambda )
    = O \left( \frac{r^2}{k} \right),  \hspace{0.5ex} \forall
    \hspace{0.5ex} r > 0
  \end{equation}
\end{lemma}

\begin{proof}
Since $a_1, \cdots, a_m$ are independent standard complex Gaussian random variables, the linear combination of them is also complex Gaussian, with independent real and imaginary part,
\begin{align}
    z := a_1 v_1 + \cdots + a_m v_m \sim \CC \calN \left(0, \sum_{i=1}^m |v_i|^2 \right) \ .
\end{align}
\cref{prop:prob_z_circle} then gives $\Pr(|z| \le r\lambda) = 1 - e^{-r^2\lambda^2/\sum_i |v_i|^2} \le 1 - e^{-r^2/k} = O(r^2/k)$.
\end{proof}

We also require Lemma 2.1 of Ref.~\cite{tao2009permanent}, which bounds the probability of many (not necessarily independent) events occurring, given that each has a non-vanishing probability.
\begin{lemma}[First moment {\cite{tao2009permanent}}]
    \label{lemma:first_moment}Let $E_1, \ldots, E_m$ be arbitrary events (not
    necessarily independent) such that $\Pr (E_i) \ge 1 - \delta$ for all
    $1 \leq i \leq m$ and some $\delta > 0$, and let $0 < c < 1$. Then,
    \begin{equation}
        \Pr ( \text{At most $c m$ among $E_1, \ldots, E_m$ are false} ) \ge 1 - \frac{\delta}{c} \, .
    \end{equation}
\end{lemma}

Unlike Ref.~\cite{tao2009permanent}, we use the McDiarmid inequality~\cite{mcdiarmid1989method} instead of Azuma's inequality~\cite{alon2016probabilistic} as our concentration tool, because our analysis involves unbounded Gaussian variables.
\begin{lemma}[McDiarmid inequality]
  \label{lemma:McDiarmid}Let $X_1, X_2, \ldots, X_n$ be independent random variables
  such that $X_i \in \mathfrak{X}_i$ for some measurable sets
  $\mathfrak{X}_i$. Suppose that $f : \prod_{i = 1}^n \mathfrak{X}_i
  \rightarrow \mathbb{R}$ is ``Lipschitz'' in the following sense:   For each $k \le n$, and any two sequences $\vb{x}, \vb{x}' \in \prod_{i = 1}^n
  \mathfrak{X}_i$ that differ only in the $k$-th coordinate, we have
  \begin{equation}
    | f (\vb{x}) - f (\vb{x}') | \le \sigma_k \, .
  \end{equation}
  Let $Y = f (X_1, X_2, \ldots, X_n)$.  Then for any $\alpha > 0$,
  \begin{equation}
    \Pr \{ | Y - \Ex [Y] | \ge \alpha \} \le 2 \exp \left\{ -
    \frac{2 \alpha^2}{\sum_{i = 1}^n \sigma_i^2} \right\} \, .
  \end{equation}
\end{lemma}

We also need the following standard concentration inequality for supermartingales with bounded increments, which is a direct consequence of the Azuma-Hoeffding inequality.
\begin{lemma}[Supermartingale concentration]
  \label{lemma:supermartingale}
  Let $\mathcal{F}_0 \subset \mathcal{F}_1 \subset \cdots \subset \mathcal{F}_m$ be a filtration, and let $W_0, W_1, \ldots, W_m$ be $\mathcal{F}_i$-measurable real-valued random variables satisfying the supermartingale property $\Ex[W_i \mid \mathcal{F}_{i-1}] \leq W_{i-1}$ for all $1 \leq i \leq m$. Assume also that $|W_i - W_{i-1}| \leq c$ for all $1 \leq i \leq m$. Then for any $\lambda \geq 0$,
  \begin{equation}
    \Pr(W_m - W_0 \geq \lambda) \leq \exp\left(-\frac{\lambda^2}{2 c^2 m}\right) \, .
  \end{equation}
\end{lemma}

\section{The Proofs}\label{proofofproposition3.12}

The proof structure follows Tao and Vu~\cite{tao2009permanent} closely.
We use the same row-exposure strategy, the same $E_{k,N,\lambda}$ framework, and the same three-step decomposition (warm-up, growth phase, endgame).
Several results carry over unchanged, as they do not depend on the entry distribution: the First Moment Lemma (\cref{lemma:first_moment}), the supermartingale concentration inequality (\cref{lemma:supermartingale}), and the existence of many complement-disjoint heavy minors in the endgame (\cref{lemma:many_heavy_minors}).
Tao and Vu noted that their result ``holds for virtually any (not too degenerate) discrete distribution''~\cite[Remark~1.6]{tao2009permanent}; however, their proof relies on the Littlewood-Offord-Erd\H{o}s inequality for discrete random signs and Azuma's inequality for bounded increments, neither of which applies to continuous, unbounded complex Gaussian entries.
The main technical contribution of the present work is the replacement of three distribution-dependent tools:
the Littlewood-Offord-Erd\H{o}s inequality is replaced by a Gaussian anti-concentration bound (\cref{littlewoodinequality});
the discrete single-step bound of~\cite[Lemma~4.1]{tao2009permanent} is replaced by an analytic argument using the rotational symmetry of the complex Gaussian distribution (\cref{lemma:large_parent_has_large_child});
and Azuma's inequality is replaced by the McDiarmid inequality (\cref{lemma:McDiarmid}), which accommodates the unbounded support of Gaussian variables at the cost of conditioning on a high-probability event.
The propositions on maintaining and growing large minors (\cref{maintainingmanylargeminors,prop:growingmanylargeminors}) follow the same combinatorial arguments as Propositions~3.2 and~3.3 in Ref.~\cite{tao2009permanent}, with constants adjusted to the Gaussian setting; their proofs are included for completeness.

We begin with the three new distribution-dependent lemmas.
The following lemma is the Gaussian counterpart to Lemma~4.1 in Ref.~\cite{tao2009permanent}; although the statements are analogous, the proof technique is different.
\begin{lemma}
    \label{lemma:large_parent_has_large_child}
    Suppose that $M \sim \calG^{n \times n}$ is sampled.
    Let $A \in \binom{[n]}{k}$ for some $1 \le k < n$, and let $i \not\in A$. 
    Assume that the submatrix $M^{(k)}$ is fixed and we expose a random row $\vec{X}_{k+1} \sim \calG^n$. Then
    \begin{equation}
    \Pr\left( | \Per ( M_{A \cup \{ i \}} ) | \ge \left| \Per (M_A) \right| \right) \ge \frac{1}{e} \, .
    \end{equation}
\end{lemma}

\begin{proof}
By co-factor expansion,
\begin{align}
    \Per ( M_{A \cup \{ i \}} ) ={}& m_{k + 1, i} \Per ( M_A ) \nonumber \\
    & + \sum_{j \in A} m_{k + 1, j} \Per ( M_{A \cup \{ i \} - \{ j \}} ) \, ,
\end{align}
where $m_{k+1,i}$ and $m_{k+1,j}$ are standard complex Gaussian random variables obeying $\mathbb{C}\mathcal{N}(0, 1)$.  
Denote the second part above as $B$,
\begin{align}
    B := \sum_{j \in A} m_{k + 1, j} \Per \left( M_{A \cup \{ i \} - \{ j \}} \right) \ .
\end{align}
Conditioned on all entries of the first $k + 1$ rows except the element $m_{k + 1, i}$, one can regard $B$ and $\Per (M_A)$ as complex constants.
Then the conditional distribution of $\Per \left( M_{A \cup \{ i \}} \right)$ is also complex normal distribution with mean $B$ and variance $\sigma^2 := | \Per (M_A) |^2$.
Let $p(x)$ denote the conditional probability density function of $\Per \left( M_{A \cup \{ i \}} \right)$, and we have
\begin{align}
    p(x) = \frac{1}{\pi \sigma^2} \exp[-\frac{1}{\sigma^2} \overline{(x- B)} (x- B)] \, .
\end{align}

Let $D := \{ |\Per(M_{A \cup \{i\}})| < \sigma \}$ be the complementary event. Then
\begin{equation}
    \Pr (D) = \int_{|x|<\sigma} p(x) \, \dd x \le \int_{|x|<\sigma} \tfrac{1}{\pi \sigma^2} e^{-|x|^2/\sigma^2} \dd x = 1 - e^{-1} \, ,
\end{equation}
where the inequality holds because the Gaussian density $p(x)$ is maximized at $x = B$, so shifting the center to the origin only increases the integral over the disk $|x| < \sigma$; equality holds iff $B = 0$. The lemma follows from $\Pr(D) + \Pr(|\Per(M_{A \cup \{i\}})| \ge |\Per(M_A)|) = 1$.
\end{proof}

Since the entries in the $(k+1)$-th row are independent, \cref{lemma:large_parent_has_large_child} immediately yields the following (cf.\ \cite[Lemma~4.2]{tao2009permanent}).

\begin{lemma}[Growing a single large minor]
\label{maintainingasinglelargeminor}
    Suppose $M \sim \calG^{n \times n}$ is sampled, $A \in \binom{[n]}{k}$ for some $1 \le k < n$, and $I \subseteq [n] - A$. Let
    \begin{equation}
        \nu(A, I) := \#\{i \in I : |\Per(M_{A\cup\{i\}})| \ge |\Per(M_A)|\} \, .
    \end{equation}
    Conditioning on $M^{(k)}$ fixed and exposing the random row $k+1$,
    \begin{equation}
        \Pr(\nu(A, I) \ge 1) \ge 1 - \eta^{-|I|} \, ,
    \end{equation}
    where $\eta := e/(e-1)$. This implies $\Pr(E_{k+1, 1, \lambda} \mid E_{k, 1, \lambda}) \ge 1 - \eta^{-(n-k)}$.
\end{lemma}

\begin{proof}
    Independence of $m_{k+1, i}$ across $i$ gives $\Pr(\nu(A, I) = 0) < (1 - 1/e)^{|I|} = \eta^{-|I|}$, hence $\Pr(\nu(A, I) \ge 1) \ge 1 - \eta^{-|I|}$. Setting $I = [n] - A$ and recalling that $E_{k,N,\lambda}$ is the event that at least $N$ of the $k \times k$ minors are $\lambda$-heavy yields $\Pr(E_{k+1, 1, \lambda} \mid E_{k, 1, \lambda}) \ge 1 - \eta^{-(n-k)}$.
\end{proof}

Using the above lemma consecutively, we can prove the following proposition.
\begin{proposition}[Step 1]\label{prop:step_1}
Let $0 < \epsilon < 1$ be some constant and let $k_0 = \lfloor \epsilon n \rfloor$. Then, $\Pr(E_{k_0,1,1})\ge 1 - \exp (-\Omega(n)) $. 
\end{proposition}

\begin{proof}
Recall that $\Pr (E_{1, 1, 1}) \ge 1 - \eta^{- n}$ and $\eta \approx 1.582$.
Apply \cref{maintainingasinglelargeminor} consecutively and we have,
\begin{align*}
  \Pr (E_{k_0, 1, 1}) & \ge \Pr (E_{1, 1, 1}) \prod_{k=1}^{k_0-1} \Pr (E_{k+1,1,1} \mid E_{k,1,1}) \\
  & \ge (1 - \eta^{-n})(1 - \eta^{-(n-k_0)})^{k_0 - 1} \\
  & \ge 1 - \exp(-\Omega(n))\, ,
\end{align*}
where the last step uses the Bernoulli inequality together with $k_0 = \lfloor\epsilon n\rfloor$.
\end{proof}

Applying the Chernoff bound to independent applications of \cref{lemma:large_parent_has_large_child} gives the following (cf.\ \cite[Lemma~4.2]{tao2009permanent}).
\begin{lemma}\label{lemma:growing_many_large_minors}
    Under the assumptions of \cref{maintainingasinglelargeminor} with $k < (1-\epsilon)n$,
    \begin{equation}\label{eq12}
        \Pr(\nu(A, I) \ge |I|/3) \ge 1 - \exp(-\Omega(|I|)) \, ,
    \end{equation}
    which implies $\Pr(E_{k+1, \epsilon n/3, \lambda} \mid E_{k,1,\lambda}) \ge 1 - \exp(-\Omega(\epsilon n))$.
\end{lemma}

\begin{proof}
Enumerate $I = \{i_j\}_{j=1}^{|I|}$ and let $X_j := \mathbf{1}[|\Per(M_{A \cup \{i_j\}})| \ge |\Per(M_A)|]$, so that $\nu(A,I) = \sum_j X_j$. By \cref{lemma:large_parent_has_large_child}, $\Pr(X_j = 1) \ge 1/e$, and independence of $a_{k+1, i_j}$ across $j$ makes the $X_j$ independent, hence $\mu := \mathbf{E}[\nu(A,I)] \ge |I|/e$. The Chernoff bound with $\delta = 1 - e/3$ gives
\begin{equation}
    \Pr(\nu(A,I) < |I|/3) \le \exp(-\mu \delta^2/2) \le \exp(-\Omega(|I|)) \, ,
\end{equation}
which is the first claim. The second follows from $|I| = n - k > \epsilon n$.
\end{proof}

The following proposition maintains the count of heavy minors as $k$ grows (cf.\ \cite[Proposition~3.2]{tao2009permanent}).
\begin{proposition}[\cref{maintainingmanylargeminors_main} restated]
\label{maintainingmanylargeminors}
    Let $1 \le k \le ( 1 - \epsilon ) n$ for some $\epsilon > 0$, let $N \ge 1$ and let $\lambda > 0$. 
    Then we have
    \begin{equation}
        \Pr ( E_{k + 1, \epsilon N / 6, \lambda} | E_{k, N, \lambda} ) \ge 1 - 2\exp ( - \Omega ( \epsilon n ) ) \ .
    \end{equation}
\end{proposition}

\begin{proof}
    Assume $E_{k, N, \lambda}$ holds: there exist $N$ $\lambda$-heavy $k$-minors $A_1, \ldots, A_N \in \binom{[n]}{k}$, each admitting at least $\epsilon n$ children since $k \le (1-\epsilon)n$.
    Let $G_i$ be the event that $A_i$ has at least $\epsilon n/3$ $\lambda$-heavy children; \cref{lemma:growing_many_large_minors} yields $\Pr(G_i) \ge 1 - \exp(-\Omega(\epsilon n))$.
    Applying \cref{lemma:first_moment} to $\{G_i\}_{i=1}^N$ with $c = 1/2$, $\delta = \exp(-\Omega(\epsilon n))$ gives that at least $N/2$ of the $G_i$ hold with probability $\ge 1 - 2\exp(-\Omega(\epsilon n))$.

    Form the bipartite graph with parents $V_P = \{A_i : G_i \text{ holds}\}$ and children $V_C = \{\lambda\text{-heavy } (k+1)\text{-minors}\}$, edges linking parent-child pairs. On the high-probability event above, $|V_P| \ge N/2$. Each vertex in $V_P$ contributes at least $\epsilon n/3$ edges and each $(k+1)$-minor has at most $n$ parents, so $|V_C| \ge |V_P|\epsilon/3 \ge \epsilon N/6$, yielding $\Pr(E_{k+1, \epsilon N/6, \lambda} \mid E_{k, N, \lambda}) \ge 1 - 2\exp(-\Omega(\epsilon n))$.
\end{proof}

The following proposition, analogous to Proposition 3.3 in Ref.~\cite{tao2009permanent}, shows that either $N$ or $\lambda$ can be improved when growing $k$ to $k+1$, at the cost of a higher failure rate. The proof follows Ref.~\cite{tao2009permanent} with constants adapted to the Gaussian distribution.

\begin{proposition}[\cref{lemma:growingmanylargeminors_main} restated]
\label{prop:growingmanylargeminors}
    Let $1 \le k \le ( 1 -  \epsilon ) n$ for some $ \epsilon > 0$, let $N \ge 1$, let $0 < c < 1$, and let $\lambda > 0$. 
    We can partition the event $E_{k, N, \lambda}$ as $E_{k, N, \lambda, c}' \vee E''_{k, N, \lambda, c}$, where the events $E_{k, N, \lambda, c}'$ and $E''_{k, N, \lambda, c}$ depend only on the first $k$ rows of $M$.
    Then, we have
    \begin{equation}
    \Pr ( E_{k + 1, n^c N, \lambda}  | E'_{k, N, \lambda, c}
    ) \ge \frac{1}{3}
    \end{equation}
    and
    \begin{equation}
    \Pr ( E_{k + 1,  \epsilon N / 4, n^{1 / 2 - c} \lambda}  |
    E''_{k, N, \lambda, c} ) \ge 1 - n^{- c / 4} \ .
    \end{equation}
\end{proposition}

\begin{proof}
    Fix $k,  \epsilon, N, c, \lambda$.
    We condition on the first $k$ rows of $M$ and assume that $E_{k, N, \lambda}$ holds, which means that we can thus find $N$ distinct $\lambda$-heavy $k$-minors labelled by $A_1, \ldots, A_N \in \binom{[n]}{k}$.
    For each $l \ge 1$, let $F_l$ denote the number of $A' \in \binom{[n]}{k+1}$ which have exactly $l$ parents in the set $\{ A_1, \ldots, A_N \}$.
    Note that $A'$ is of the form $A_i \cup \{j\}$.
    Now, consider a bipartite graph consisting of two sets of vertices $V_P$ and $V_C$, where $V_P = \{A_1, \ldots, A_N\}$ and $V_C = \{ A_i \cup \{j\} : i \in [N], j \in [n]-A_i \}$.
    An edge is drawn between $A_i \in V_P$ and $A' = A_i \cup \{j\} \in V_C$.    
    Each $A'$ has at most $k+1 < n$ parents in $\{ A_1, \ldots, A_N \}$, which implies that
    \begin{equation}
        \# \mathrm{edges} = \sum_{l = 1}^n l F_l \ ,
    \end{equation}
    where $F_l = 0$ for $l > k+1$.
    On the other hand, since each $A_j$ has at least $ \epsilon n$ children, we have
    \begin{equation}
        \# \mathrm{edges} \ge  \epsilon n N \ .
    \end{equation}
    Combining these, we have
    \begin{equation}
        \sum_{l = 1}^n l F_l \ge  \epsilon n N \ .
    \end{equation}

    Set $K := \lfloor (\epsilon \gamma) n^{1-c} \rfloor$ with $\gamma := \frac{1}{2}(1 - \frac{3(e-1)}{2e}) \approx 0.026$. Bounding the coefficient $l$ by $K$ for $l \le K$ and by $n$ otherwise gives $K(F_1 + \cdots + F_K) + n(F_{K+1} + \cdots + F_n) \ge \sum_l l F_l \ge \epsilon n N$, so one of the following holds:
    \begin{equation}\label{eq:event_E'}
        F_1 + \cdots + F_K \ge \frac{ \epsilon n N}{2K}
    \end{equation}
    or
    \begin{equation}\label{eq:event_E''}
        F_{K + 1} + \cdots + F_n \ge \frac{ \epsilon N}{2} \ .
    \end{equation}
    We then let $E_{k, N, \lambda, c}'$ be the event that \cref{eq:event_E'} holds and $E_{k, N, \lambda, c}''$ be the event that \cref{eq:event_E'} fails but \cref{eq:event_E''} holds.

    \emph{Case $E'_{k,N,\lambda,c}$.}
    By \cref{eq:event_E'} there are at least $\epsilon n N/(2K)$ elements $A' \in \binom{[n]}{k+1}$, each with at least one parent in $\{A_1,\ldots,A_N\}$; by \cref{lemma:large_parent_has_large_child} each such $A'$ is $\lambda$-heavy with probability $\ge 1/e$. Applying \cref{lemma:first_moment} with $(m,\delta,c) = (\epsilon n N/(2K), (e-1)/e, 3(e-1)/(2e))$ and using $(\epsilon n N)/(2K)(1 - 3(e-1)/(2e)) \ge n^c N$ gives
    \begin{equation}
        \Pr(\text{at least } n^c N \text{ such } A' \text{ are } \lambda\text{-heavy} \mid E'_{k,N,\lambda,c}) \ge \tfrac{1}{3} \, .
    \end{equation}

    \emph{Case $E''_{k,N,\lambda,c}$.}
    By \cref{eq:event_E''} there are at least $\epsilon N/2$ elements $A' \in \binom{[n]}{k+1}$, each with at least $K$ parents in $\{A_1,\ldots,A_N\}$. Cofactor expansion gives $\Per(M_{A'}) = a_1 v_1 + \cdots + a_{k+1} v_{k+1}$ with $a_i \sim \calG$ independent and at least $K$ of the $|v_i| \ge \lambda$, so \cref{littlewoodinequality} yields $\Pr(|\Per(M_{A'})| \le n^{1/2-c}\lambda) \le O(n^{1-2c}/K)$. \cref{lemma:first_moment} then gives, for $n > (1/\epsilon)^{4/(3c)}$, that at least $\epsilon N/4$ of the $A'$ are $(n^{1/2-c}\lambda)$-heavy with probability $\ge 1 - n^{-c/4}$.
\end{proof}

Based on the above propositions, we now describe the algorithmic construction for Step 2 and prove the growth-phase result.
The construction follows the same strategy as the one on page 662 of Ref.~\cite{tao2009permanent}, adapted for random Gaussian matrices using \cref{maintainingmanylargeminors,prop:growingmanylargeminors}.

Fix a desired $\alpha > 0$, choose $\epsilon > 0$ sufficiently smaller than $\alpha$, and choose $\epsilon' > 0$ with $\frac{1}{6}\epsilon > \epsilon' > 0$ sufficiently small compared to $\epsilon$.
Let $k_0 = \floor{\epsilon n}$ and $k_1 = \floor{(1-\epsilon)n}$.
Starting from $E_{k_0, 1, 1}$ (established by \cref{prop:step_1}), we grow $k$ from $k_0$ to $k_1$ while tracking the state $(N_k, \lambda_k)$, initialized with $N_{k_0} = 1$ and $\lambda_{k_0} = 1$.

\paragraph{The four types of steps.}
At each step $k \to k+1$ (for $k_0 \leq k < k_1$), we first apply \cref{maintainingmanylargeminors} to ensure that $E_{k+1, \epsilon N_k / 6, \lambda_k}$ holds, which succeeds with probability at least $1 - 2\exp(-\Omega(\epsilon n))$.
Conditioned on this success, we then apply \cref{prop:growingmanylargeminors} with $c = \epsilon'$ to attempt either a population increase or a magnitude increase.
Depending on which events occur, the step falls into one of four types:

\begin{itemize}
    \item \textbf{Type I (Population growth):} The event $E'_{k, N_k, \lambda_k, \epsilon'}$ holds and $E_{k+1, n^{\epsilon'} N_k, \lambda_k}$ occurs. We set $N_{k+1} := n^{\epsilon'} N_k$ and $\lambda_{k+1} := \lambda_k$.

    \item \textbf{Type II (Maintenance under population regime):} The event $E'_{k, N_k, \lambda_k, \epsilon'}$ holds but the population growth $E_{k+1, n^{\epsilon'} N_k, \lambda_k}$ fails. By \cref{maintainingmanylargeminors}, we still have $E_{k+1, \epsilon N_k / 6, \lambda_k}$. We set $N_{k+1} := \epsilon N_k / 6$ and $\lambda_{k+1} := \lambda_k$.

    \item \textbf{Type III (Maintenance under magnitude regime):} The event $E''_{k, N_k, \lambda_k, \epsilon'}$ holds but the magnitude growth $E_{k+1, \epsilon N_k / 4, n^{1/2 - \epsilon'} \lambda_k}$ fails. By \cref{maintainingmanylargeminors}, we still have $E_{k+1, \epsilon N_k / 6, \lambda_k}$. We set $N_{k+1} := \epsilon N_k / 6$ and $\lambda_{k+1} := \lambda_k$.

    \item \textbf{Type IV (Magnitude growth):} The event $E''_{k, N_k, \lambda_k, \epsilon'}$ holds and $E_{k+1, \epsilon N_k / 4, n^{1/2 - \epsilon'} \lambda_k}$ occurs. We set $N_{k+1} := \epsilon N_k / 4$ and $\lambda_{k+1} := n^{1/2 - \epsilon'} \lambda_k$.
\end{itemize}

If none of the above holds (i.e., the maintaining step also fails), we declare the algorithm to have failed at step $k$.
The algorithm outputs ``success'' if it completes all $k_1 - k_0$ steps without failure and $\lambda_{k_1} \geq n^{(1/2 - \alpha)n}$.

\paragraph{The potential function.}
To analyze the algorithm, we introduce a potential function $W_k$ for $k_0 \leq k \leq k_1$, defined recursively by $W_{k_0} := 0$ and
\begin{equation}
\label{eq:potential_Wk}
    W_{k+1} := W_k + \left(1 - \frac{\epsilon}{2}\right) - 3\, \mathbf{I}_{k, \text{type I}} - \mathbf{I}_{k, \text{type IV}} \, ,
\end{equation}
where $\mathbf{I}_{k, \text{type I}}$ and $\mathbf{I}_{k, \text{type IV}}$ are the indicator functions that step $k$ is of Type I or Type IV, respectively. Since Types II and III contribute $0$ to the indicator sum, their increments are $(1 - \epsilon/2)$.

\begin{proposition}[Success probability]\label{prop:success_prob}
    The probability for the algorithm to be successful is at least $1 - \exp(-\Omega(\epsilon' n))$.
\end{proposition}

\begin{proof}
The algorithm succeeds if two conditions hold simultaneously:
\begin{enumerate}
    \item[(I)] Every step $k$ from $k_0$ to $k_1$ completes without failure (i.e., the maintaining step always succeeds).
    \item[(II)] $W_{k_1} \leq \epsilon' n / 2$.
\end{enumerate}

\emph{Condition I.}
At $k = k_0$, the failure probability is $\exp(-\Omega(n))$ by \cref{prop:step_1}.
At each subsequent step $k_0 < k \leq k_1$, the only way to fail is if $E_{k+1, \epsilon N_k/6, \lambda_k}$ does not hold given $E_{k, N_k, \lambda_k}$. By \cref{maintainingmanylargeminors}, this failure probability is at most $2\exp(-\Omega(\epsilon n))$.
Taking a union bound over all $k_1 - k_0 + 2 \leq n$ steps:
\begin{align}
    \Pr(\text{Condition I holds})
    &\geq \left[1 - 2\exp(-\Omega(\epsilon n))\right]^n \nonumber \\
    &\geq 1 - \exp(-\Omega(\epsilon' n)) \, .
\end{align}

\emph{Condition II.}
We show that $W_k$ satisfies the supermartingale-type property
\begin{equation}
    \Ex[W_{k+1} \mid M^{(k)}] \leq W_k \, .
\end{equation}
Indeed, expanding the definition gives
\begin{align}
    W_k - \Ex[W_{k+1} \mid M^{(k)}]
    &= 3\, \Ex[\mathbf{I}_{k, \text{type I}} \mid M^{(k)}] \nonumber \\
    &\quad + \Ex[\mathbf{I}_{k, \text{type IV}} \mid M^{(k)}] \nonumber \\
    &\quad - \left(1 - \tfrac{\epsilon}{2}\right) \, .
\end{align}
Given $M^{(k)}$, exactly one of the two regimes $E'_{k, N_k, \lambda_k, \epsilon'}$ or $E''_{k, N_k, \lambda_k, \epsilon'}$ holds. We verify the inequality in each case.

\emph{Case: magnitude regime ($E''_k$).}
In this case, Type I cannot occur, so $\mathbf{I}_{k, \text{type I}} = 0$.
By \cref{prop:growingmanylargeminors}, $\Pr(\text{Type IV step} \mid E''_k) \geq 1 - n^{-\epsilon'/4}$.
Therefore:
\begin{align}
    W_k - \Ex[W_{k+1} \mid M^{(k)}]
    &= 0 + (1 - n^{-\epsilon'/4}) - 1 + \tfrac{\epsilon}{2} \nonumber \\
    &= \tfrac{\epsilon}{2} - n^{-\epsilon'/4} \geq 0
\end{align}
for $n$ sufficiently large.

\emph{Case: population regime ($E'_k$).}
In this case, Type IV cannot occur, so $\mathbf{I}_{k, \text{type IV}} = 0$.
By \cref{prop:growingmanylargeminors}, $\Pr(\text{Type I step} \mid E'_k) \geq 1/3$.
Therefore:
\begin{align}
    W_k - \Ex[W_{k+1} \mid M^{(k)}] &\geq 3 \cdot \frac{1}{3} + 0 - 1 + \frac{\epsilon}{2} = \frac{\epsilon}{2} \geq 0 \, .
\end{align}
Moreover, $|W_{k+1} - W_k| \leq 3$ deterministically (since each indicator is either 0 or 1).
Applying \cref{lemma:supermartingale} with $c = 3$ and $m = k_1 - k_0$:
\begin{equation}
    \Pr\left(W_{k_1} \geq \frac{\epsilon' n}{2}\right) \leq \exp\left(-\Omega(\epsilon'^2 n)\right) \, .
\end{equation}
Combining Conditions I and II completes the proof.
\end{proof}

\begin{proposition}[Success result]\label{prop:success_result}
    If the algorithm is successful, then $E_{k_1, 1, n^{(1/2 - \alpha)n}}$ holds.
\end{proposition}

\begin{proof}
Denote $S_I := \sum_{k=k_0}^{k_1-1} \mathbf{I}_{k, \text{type I}}$ and $S_{IV} := \sum_{k=k_0}^{k_1-1} \mathbf{I}_{k, \text{type IV}}$. If the algorithm succeeds, then $W_{k_1} = (k_1 - k_0)(1 - \epsilon/2) - 3 S_I - S_{IV} \le \epsilon' n/2$, i.e.,
\begin{equation}
    3 S_I + S_{IV} \geq (k_1 - k_0)\left(1 - \frac{\epsilon}{2}\right) - \frac{\epsilon' n}{2} \, .
\end{equation}

The number of Type I steps $S_I$ is $o(n)$. Indeed, each Type I step multiplies $N_k$ by a factor of $n^{\epsilon'}$, while any other step decreases $N_k$ by at most a constant factor $\epsilon/6$. Since $N_k \leq \binom{n}{k} \leq 2^n$ for all $k$, we have $S_I \leq \frac{n \log 2}{\epsilon' \log n} = o(n)$.

Therefore $S_{IV} \ge (k_1 - k_0)(1 - \epsilon/2) - \epsilon' n/2 - 3 S_I \ge (1-2\epsilon)(1-\epsilon/2)n - \epsilon' n/2 - o(n) = (1 - 5\epsilon/2 + \epsilon^2)n - \epsilon' n/2 - o(n) \ge (1-3\epsilon)n$ for $\epsilon, \epsilon'$ sufficiently small.
Since each Type IV step multiplies $\lambda_k$ by $n^{1/2 - \epsilon'}$, starting from $\lambda_{k_0} = 1$:
\begin{equation}
    \lambda_{k_1} \geq \left(n^{1/2-\epsilon'}\right)^{(1-3\epsilon)n} = n^{(1/2-\epsilon')(1-3\epsilon)n} \, .
\end{equation}
Since $(1/2 - \epsilon')(1-3\epsilon) \geq 1/2 - 3\epsilon/2 - \epsilon' \geq 1/2 - \alpha$ for $\epsilon$ and $\epsilon'$ sufficiently small relative to $\alpha$, we conclude that $\lambda_{k_1} \geq n^{(1/2-\alpha)n}$.
\end{proof}

Combining \cref{prop:success_prob,prop:success_result}, we obtain the main growth-phase result:
\begin{proposition}\label{proposition:second_step}
Fix a desired $\alpha > 0$ and choose $\epsilon > 0$ sufficiently smaller than $\alpha$. Let $k_1 = \floor{(1 - \epsilon) n}$. Then, we have
    \begin{equation}
        \Pr(E_{k_1,1, n ^{(1/2 - \alpha)n}}) \geq 1 - \exp(-\Omega(\epsilon' n)) \, ,
    \end{equation}
where $\epsilon' > 0$ is sufficiently small compared to $\epsilon$.
\end{proposition}

The third step is to deal with the last few values of $k$ such that $(1-\epsilon)n \le k \le n$. 
Below, we prove a proposition that with vanishing failure probability, grows $k$ from $k_1$ to $n$ while only drops the magnitude of the permanent of minors by a factor of $n^{- \log n}$. This is called the endgame (proposition 3.4) in Ref.~\cite{tao2009permanent}. We will prove that this endgame proposition also holds for random Gaussian matrices in \cref{sec:endgame}.
\begin{proposition}[Endgame]
  \label{endgame}Let $1 \le k \le ( 1 -  \epsilon )
  n$ for some $ \epsilon > 0$, and let $\lambda > 0$. Then
  \begin{equation}
    \Pr ( E_{n, 1, n^{- \log n}  \lambda}  | E_{k, 1, \lambda}
    ) \ge 1 - n^{- \Omega (1)} \, .
  \end{equation}
\end{proposition}

Applying the endgame (\cref{endgame}) to \cref{proposition:second_step} with $\lambda = n^{(1/2 - \alpha)n}$ chains the bounds: $\Pr(E_{n,1,n^{(1/2-\alpha)n - \log n}}) \ge (1 - n^{-\Omega(1)})(1 - e^{-\Omega(\epsilon' n)}) = 1 - O(n^{-\Omega(1)})$. Since $(1/2 - \alpha)n - \log n = (1/2 - \alpha + o(1))n$, this establishes \cref{thm:weakPACC}.

\section{Proof of the endgame}\label{sec:endgame}

In this section, we prove the endgame proposition (\cref{endgame}) for Gaussian random matrices. The proof follows Ref.~\cite{tao2009permanent}, with the McDiarmid inequality replacing Azuma's inequality as discussed in \cref{proofofproposition3.12}.

Fix $k, n, \lambda$. We condition on $M^{(k)}$ and assume that $E_{k, 1, \lambda}$
holds, thus there is one of the elements of $\binom{[n]}{k}$ that is $\lambda$-heavy. By symmetry, we may assume without loss of generality that $[k]$ is
$\lambda$-heavy. Our goal is to show that $[n]$ is
$n^{- \log n}  \lambda$-heavy with probability $1 - n^{-
\Omega (1)}$. Set $L := \frac{1}{100} \lfloor \log n \rfloor$. We first show that there
are plenty of heavy minors in $\binom{[n]}{n-L}$.

The following  lemma is proved in Ref.~\cite{tao2009permanent}, which is a result about independent events and does not depend on the distribution of the entries at all.
\begin{lemma}[Many heavy minors of order $n-L$]
  Let $B \subset \binom{[n]-[k]}{2L}$. Then with probability $1 - \exp ( - \Omega (L)
   )$, there exists a $\lambda$-heavy minor $A \in \binom{[n]}{n-L}$ which contains $[n] - B$.\label{lemma:many_heavy_minors}
\end{lemma}

The following corollary is proved in Ref.~\cite{tao2009permanent} as Corollary 6.2, whose proof is based on \cref{lemma:many_heavy_minors} and the lemma of first moment (\cref{lemma:first_moment}).
\begin{corollary}[Many complement-disjoint heavy minors of order $n-L$]
  We have
  \begin{equation}
    \Pr ( F_{L, \left\lfloor  \epsilon n / 10 L \right\rfloor, \lambda}
    ) = 1 - \exp ( - \Omega (L) ) = 1 - n^{- \Omega (1)} \, .
  \end{equation}
\end{corollary}

The following lemma is slightly different from the corresponding one in Ref.~\cite{tao2009permanent} and part of its proof is different.
For any integer $N \ge 1$, any $1 \le j \le L$, and any
$\lambda' > 0$, let $F_{j, N, \lambda'}$ denote the event that there exists
$N$ $\lambda'$-heavy minors $A_1, \ldots, A_N \in \binom{[n]}{n-j}$ whose complements $[n] - A_1, [n] - A_2, \ldots, [n] -
A_N$ are disjoint.

\begin{lemma}
  \label{iteration}Let $1 < j \le L$, $N \ge n^{0.5}$ (say), and
  $\lambda' > 0$. Then
  \begin{equation}
    \Pr ( F_{j - 1, \lfloor N / 12 \rfloor, \lambda' / n^2}  |
    F_{j, N, \lambda'} ) \ge 1 - n^{- \Omega (1)}
  \end{equation}
\end{lemma}

\begin{proof}
Fix $j, N$. We condition on $M^{(n - j)}$ so that $F_{j, N, \lambda'}$ hold.
Thus we can find $\lambda'$-heavy sets $A_1, \ldots, A_N \in \binom{[n]}{n-j}$ with disjoint complements, which we now fix. For each
$A_i$, we arbitrarily choose a child $B_i = A_i  \cup \{ h_i \} \in \binom{[n]}{n-j+1}$. By construction, the $B_1, \ldots, B_N$ also have
disjoint complements and the $h_i$ are different.

Let $T := \lfloor n^{0.1} \rfloor$. Call a child $B_i$ good if it has at
least $T$ $\lambda' / n^2$-heavy parents of which $A_i$ will be one of them,
and bad otherwise. There are two cases to consider.

\textbf{Case 1: At least half of the $B_i$ are good.} For this first case, it is assumed that at least half of the $B_i$'s are
good. Apply the co-factor expansion and the \cref{littlewoodinequality}
with $\lambda = \frac{\lambda'}{n^2}, r = 1, m = n - j, k = T$ and it can be concluded that each good $B_i$ is NOT $\lambda' / n^2$-heavy
with probability at most $O ( \frac{1}{T} )$,

So each good $B_i$ is $\lambda' / n^2$-heavy with probability at least $1 - O
( \frac{1}{T} )$. Among the good $B_i$, apply the first moment \cref{lemma:first_moment} with $  \delta = O ( \frac{1}{T} ), m = \frac{N}{2}, c = \frac{1}{6}$, then one gets
\begin{equation}
  \Pr(\#\{B_i : \lambda'/n^2\text{-heavy}\} \ge \lfloor N/12 \rfloor) \ge 1 - n^{-\Omega(1)} \, .
\end{equation}
This proves the claim.

\textbf{Case 2: At least half of the $B_i$ are bad.}  Let $I$ be the set of all $i$ such that
$B_i$ is bad and $H$ be the set of $h_i$, $i \in I$. Draw a bipartite graph
$G$ between $I$ and $H$ by connecting $i$ to $h_{i'}$ if and only if $B_i -
\{ h_{i'} \}$ is $\lambda' / n^2$-heavy. By definition, $B_i - \{ h_{i'} \}$ is a parent of $B_i$. Since each $B_i$
are presumed bad, each $i \in I$ has degree at most $T$. Thus, $\sum_{i \in I} \deg (h_i) =\# \mathrm{edges} \le | I | T $.
Define $I' = \{ i \in I \mid \deg (h_i) \le 2 T \}$. Since every $i \in I \setminus I'$ contributes more than $2T$ to $\sum_{i \in I} \deg(h_i)$, we have $|I|T \ge 2T(|I| - |I'|)$, hence $|I'| \ge |I|/2 \ge N/4$.

Following Ref.~\cite{tao2009permanent}, we condition on the entries of the $n - j + 1$ row not in the columns
determined by $I'$, For each $i \in I'$, let
$  Y_i := \min\{ \frac{| \Per M_{B_i} |}{\lambda'}, 1 \} $. By \cref{lemma:large_parent_has_large_child}, $\Pr ( \left| \Per ( M_{A \cup \{ i \}} )
\right| \ge | \Per (M_A) | ) \ge \frac{1}{e}$, and each
$A_i$ has been assumed to be $\lambda'$-heavy, so $\Pr ( | \Per
(M_{B_i}) | \ge \lambda' ) \ge \frac{1}{e}$. Therefore,
\begin{equation}
    \mathbf{E}(Y_i) \ge \Pr ( | \Per (M_{B_i}) | \ge \lambda' ) \ge \tfrac{1}{e} \, .
\end{equation}
Define $Y= \sum_{i \in I'} Y_i$, then by the linearity of expectation, $\mathbf{E} (Y) \ge \frac{| I'
|}{e} \ge  \frac{N}{4 e} $.

The goal is to apply the McDiarmid inequality (\cref{lemma:McDiarmid}) to bound the probability of $\Pr(|Y-\mathbf{E}(Y)|\ge |I'|/1000)$. To this end, we estimate the effect of each random entry $a_{n - j + 1, h}$ on $Y$. This estimation is new compared to Ref.~\cite{tao2009permanent}.

Let $\{ i_l \}_{l = 1}^{| I' |}$ be an enumeration of $I'$, and set $x_l := a_{n - j + 1, h_{i_l}}$. Since $x_l \sim \mathcal{N}(0, 1)_\mathbb{C}$, a union bound gives $\Pr(\forall l,\, |x_l| \le n) \ge 1 - |I'| e^{-n^2} = 1 - o(1)$.

Conditioned that all $| x_l | \le n$, we will prove that
\begin{multline}
  | Y (x_1, x_2, \ldots, x_l', \ldots, x_{| I' |}) \\
  - Y (x_1, x_2, \ldots, x_l, \ldots, x_{| I' |}) | \le 3 T \, .
\end{multline}
Each coordinate $a_{n-j+1,h}$ affects $Y_i$ only when $h \in B_i$, with change at most $1$ (conditioned on $|a_{n-j+1,h}| < n$, the bound is $2/n$ when the corresponding $(n-j)$-minor is not $\lambda'/n^2$-heavy). By the definition of $I'$, at most $2T$ indices $i \in I'$ satisfy $h \in B_i$, so $|Y' - Y| \le 2T + (2/n)\cdot n \le 3T$. Applying \cref{lemma:McDiarmid} yields
\begin{equation}
  \Pr (| Y - \mathbf{E} (Y) | \ge | I' | / 1000) \le 2 e^{-\Omega(|I'|/T^2)} = n^{-\Omega(1)} \, ,
\end{equation}
where we have used the assumption that $N \ge n^{0.5} \ge T^2 n^{0.1}$.

Combined with $\mathbf{E}(Y) \ge N/(4e)$, this gives $\Pr(Y \ge N/11) \ge 1 - n^{-\Omega(1)}$.
Conditioned on $Y \ge N / 11$, suppose there are $z$ indices satisfying $\frac{|
\Per M_{B_i} |}{\lambda'} \ge \frac{1}{n^2}$. Then there are $| I'
| - z$ indices satisfying $\frac{| \Per M_{B_i} |}{\lambda'} < \frac{1}{n^2}$, this
means,
\begin{equation}
    \frac{N}{11} \le Y \le z \times 1 + (| I' |
   - z)  \frac{1}{n^2} \le z +
  \frac{N}{ n^2} \, ,
\end{equation}
which implies $z \ge \frac{N}{12}$. This
concludes the proof of \cref{iteration}.
\end{proof}

Iterating \cref{iteration} for $L \le \frac{\log n}{100}$ steps as in Ref.~\cite{tao2009permanent}
starting with $\Pr ( F_{L, \left\lfloor  \epsilon n / 10 L
\right\rfloor, \lambda} ) = 1 - \exp ( - \Omega (L) ) = 1 -
n^{- \Omega (1)}$, we conclude that
\begin{equation}
  \Pr ( F_{1, \lfloor n^{0.5} \rfloor, n^{- \log n}  \lambda} )
  \ge 1 - n^{- \Omega (1)}
\end{equation}
Now suppose that $F_{1, \lfloor n^{0.5} \rfloor, n^{- \log n} \lambda}$
holds, thus there are at least $\lfloor n^{0.5} \rfloor$ $(n^{- \log n}
\lambda)$-heavy minors in $\binom{[n]}{n-1}$. Applying the co-factor expansion and \cref{littlewoodinequality}, we conclude that $[n]$ is $n^{- \log n} 
\lambda$-heavy with probability at least $1 - O \left( \frac{1}{n^{0.5}} \right) = 1 - n^{- \Omega (1)}$. This concludes the proof of the endgame (\cref{endgame}).

\bibliographystyle{IEEEtran}
\bibliography{references.bib}

\end{document}